\documentclass[leqno]{amsart}
\usepackage{amsmath,amsthm,amsfonts,amssymb,amstext}

\usepackage{setspace}

\numberwithin{equation}{section}
\theoremstyle{plain}

\newtheorem{theorem}{Theorem}

\newtheorem{lemma}[theorem]{Lemma}

\newtheorem{remark}[theorem]{Remark}

\begin{document}

\title[Optimal Strategies for a Long-Term Static Investor]{Optimal Strategies for a Long-Term Static Investor}

\author{LINGJIONG ZHU}
\address
{Courant Institute of Mathematical Sciences\newline
\indent New York University\newline
\indent 251 Mercer Street\newline
\indent New York, NY-10012\newline
\indent United States of America}
\email{ling@cims.nyu.edu}

\date{21 May 2014. \textit{Revised:} 21 May 2014}
%Finance, portfolios, investment;Stochastic models;Applications of stochastic analysis. 

\begin{abstract}
The optimal strategies for a long-term static investor are studied. Given a portfolio of a stock and a bond, we derive
the optimal allocation of the capitols to maximize the expected long-term growth rate of a utility function of the wealth.
When the bond has constant interest rate, three models for the underlying stock price processes are studied: Heston model, 3/2 model
and jump diffusion model. We also study the optimal strategies for a portfolio in which the stock price process follows a Black-Scholes
model and the bond process has a Vasicek interest rate that is correlated to the stock price.
\end{abstract}

\maketitle

\section{Introduction}

In this article, we are interested in the long-term optimal strategies for a static investor.
The investor starts with a known initial wealth $V_{0}>0$ and the wealth at time $t$ is denoted
by $V_{t}$. The investor decides what fraction of wealth $\alpha_{t}$ to invest in a stock $S_{t}$ and the remaining $1-\alpha_{t}$
in a bond $r_{t}$, i.e.
\begin{equation}
\frac{dV_{t}}{V_{t}}=\alpha_{t}\frac{dS_{t}}{S_{t}}+(1-\alpha_{t})r_{t}dt.
\end{equation}

For a static investor, we assume that $\alpha_{t}\equiv\alpha$ is a constant between $0$ and $1$, i.e. $\alpha\in[0,1]$.

We consider a hyperbolic absolute risk aversion (HARA) utility function $u(c)$
with constant relative risk aversion coefficient $\gamma\in(0,1)$, i.e.
\begin{equation}
u(c)=\frac{c^{1-\gamma}}{1-\gamma},\quad 0<\gamma<1.
\end{equation}

We are interested in the optimal strategy to maximize the long-term growth rate, i.e.
\begin{equation}
\max_{0\leq\alpha\leq 1}\lim_{t\rightarrow\infty}\frac{1}{t}\log E[u(V_{t})]=\max_{0\leq\alpha\leq 1}\Lambda(\alpha),
\end{equation}
if the long-term growth rate $\Lambda(\alpha):=\lim_{t\rightarrow\infty}\frac{1}{t}\log E[u(V_{t})]$ exists
for any $0\leq\alpha\leq 1$. 
For the convenience of notation, let $\theta:=1-\gamma\in(0,1)$ and thus we are interested in
\begin{equation}
\max_{0\leq\alpha\leq 1}\lim_{t\rightarrow\infty}\frac{1}{t}\log E[(V_{t})^{\theta}].
\end{equation}

The optimal long-term growth rate of expected utility of wealth has been well studied in the literatures.
Usually, the optimal strategy is taken to be dynamic and some dynamic programming equations are studied, see e.g. Fleming and Sheu \cite{Fleming}.
In this article, we only concentrate on the static strategies for the simplicity. This set-up allows us
to gain analytical tractability for some more sophisticated models like Heston model and 3/2 model. 

The problem of maximizing the long-term expected utility is closely related to maximizing the probability
that the wealth exeeds a given benchmark for large time horizon, i.e. 
$\max_{0\leq\alpha\leq 1}\lim_{t\rightarrow\infty}\frac{1}{t}\log P(V_{t}\geq V_{0}e^{xt})$, where $x$ is a given benchmark. 
In a static framework, an asymptotic outperformance
criterion was for example considered in Stutzer \cite{Stutzer}. An asymptotic dynamic version
of the outperformance managment criterion was developed by Pham \cite{Pham}. To find the optimal strategy
for the long-term growth rate, the first step is to compute the limit $\lim_{t\rightarrow\infty}\frac{1}{t}\log E[u(V_{t})]$.
Under most of the standard models, the wealth process $V_{t}$ has exponential growth rate and the existence
of a logarithmic moment generating function plus some additional conditions can be used to obtain a large deviation
principle for the probability that the wealth process outperformances a given benchmark, which is exponentially small.
The connection is provided by G\"{a}rtner-Ellis theorem, see e.g. Dembo and Zeitouni \cite{Dembo}. For a survey
on the applications of large deviations to finance, we refer to Pham \cite{PhamII}.

In this article, we study in detail the optimal strategies for a static investor investing in stocks
of two stochastic volatility models, 
i.e. the Heston model (Section \ref{HSection}), 
the 3/2 model (Section \ref{ThreeHalvesSection}) 

The Heston model, introduced by Heston \cite{Heston} is a widely used
stochastic volatility model. The volatility process is
itself a Cox-Ross-Ingersoll process, which is an affine model and
has great analytical tractability. The 3/2 model is another popular
model of stochastic volatility. It has been applied to interest
rate modeling, e.g. Ahn and Gao \cite{Ahn}. 
Carr and Sun \cite{Carr} used the 3/2 model to price variance swaps
and Drimus \cite{Drimus} used it to price options on realized variance.

Next, we study the optimal long-term static investment strategies
when the underlying stock process follows a jump diffusion model 
(Section \ref{JSection}) 
assuming the alternative investment bond has constant short-rate.
Finally, we study the case 
when the stock follows a classical Black-Scholes model 
while the bond has a Vasicek interest rate (Section \ref{VSection}).

As an illustration, let us first consider a toy model. 
Assume that the stock price follows a geometric Brownian motion with constant drift $\mu>0$ and constant
volatility $\sigma>0$ and the bond has constant short-rate $r>0$. We can write down 
a stochastic differential equation for the wealth process $V_{t}$,
\begin{equation}
dV_{t}=\alpha\mu V_{t}dt+\alpha\sigma V_{t}dB_{t}+(1-\alpha)rV_{t}dt,
\end{equation}
where $B_{t}$ is a standard Brownian motion starting at $0$ at time $0$ and therefore
\begin{equation}
V_{t}=V_{0}\exp\left\{\left(\alpha\mu+(1-\alpha)r-\frac{1}{2}\alpha^{2}\sigma^{2}\right)t+\alpha\sigma B_{t}\right\}.
\end{equation}
Hence, we can compute that
\begin{equation}
E[(V_{t})^{\theta}]=V_{0}^{\theta}e^{\theta(\alpha\mu+(1-\alpha)r-\frac{1}{2}\alpha^{2}\sigma^{2})t}e^{\frac{1}{2}\theta^{2}\alpha^{2}\sigma^{2}t}.
\end{equation}
Therefore, we are interested to maximize
\begin{equation}
\Lambda(\alpha)=\theta\left[\alpha\mu+(1-\alpha)r-\frac{1}{2}\alpha^{2}\sigma^{2}\right]+\frac{1}{2}\theta^{2}\alpha^{2}\sigma^{2},
\quad
0\leq\alpha\leq 1.
\end{equation}
It is easy to compute that
\begin{equation}
\Lambda'(\alpha)=(\theta\mu-\theta r)-(\theta-\theta^{2})\sigma^{2}\alpha.
\end{equation}
Hence, $\Lambda'(\alpha)=0$ if $\alpha=\frac{\mu-r}{(1-\theta)\sigma^{2}}$. Therefore, the optimal $\alpha^{\ast}$
is given by
\begin{equation}
\alpha^{\ast}=
\begin{cases}
0 &\text{if $\mu\leq r$},
\\
\frac{\mu-r}{(1-\theta)\sigma^{2}} &\text{if $0<\frac{\mu-r}{(1-\theta)\sigma^{2}}<1$},
\\
1 &\text{if $\mu-r\geq(1-\theta)\sigma^{2}$}.
\end{cases}
\end{equation}
The financial interpretation is clear. When $\mu\leq r$, it is optimal to invest in the bond only because the yield of the 
bond $r$ exceeds the mean return of the stock. When $\mu>r$, it is not always optimal in only invest in stocks. 
The reason is although the mean return of the stock exceeds the yield of the bond, stocks are volatile and a large volatility
can decrease the expected utility of the portfolio. This is consistent with the mean-variance analysis, which says
that given the mean, the investor has the incentive to minimize the variance.

In general, for any wealth process $V_{t}$, assume $\Lambda(\alpha)$ exists and is smooth and strictly concave,
If $\Lambda'(0)\leq 0$, then the optimal $\alpha^{\ast}$ is given by $\alpha^{\ast}=0$. Otherwise,
$\Lambda(\alpha)$ achieves a unique maximum 
at some $\alpha^{\dagger}\in(0,\infty)$.  
Then the optimal $\alpha^{\ast}$ is given by
\begin{equation}
\alpha^{\ast}
=
\begin{cases}
1 &\text{if $\alpha^{\dagger}\geq 1$},
\\
\alpha^{\dagger} &\text{if $\alpha^{\dagger}\in(0,1)$}.
\end{cases}
\end{equation}
This is the general method behind analyzing all the models in this article.

\section{Heston Model}\label{HSection}

Let us assume that the stock price follows a Heston model, namely, the stock price
has a stochastic volatility which follows a Cox-Ingersoll-Ross process,
\begin{equation}
\begin{cases}
dS_{t}=\mu S_{t}dt+\sqrt{\nu_{t}}S_{t}dB_{t},
\\
d\nu_{t}=\kappa(\gamma-\nu_{t})dt+\delta\sqrt{\nu_{t}}dW_{t},
\end{cases}
\end{equation}
where $W_{t}$ and $B_{t}$ are two standard Brownian motions and $\langle W,B\rangle_{t}=\rho t$, where
$-1\leq\rho\leq 1$ is the correlation. 
Assume that $\mu,\kappa,\gamma,\delta>0$.
The volatility process $\nu_{t}$ is a Cox-Ingersoll-Ross process, introduced by Cox et al. \cite{Cox}. 
We assume the Feller condition $2\kappa\gamma>\delta^{2}$ holds so that $\nu_{t}$ is always positive, see e.g. Feller \cite{Feller}.

The wealth process satisfies
\begin{equation}
\begin{cases}
dV_{t}=\alpha\mu V_{t}dt+\alpha\sqrt{\nu_{t}}V_{t}dB_{t}+(1-\alpha)rV_{t}dt,
\\
d\nu_{t}=\kappa(\gamma-\nu_{t})dt+\delta\sqrt{\nu_{t}}dW_{t}.
\end{cases}
\end{equation}
Then, we have
\begin{equation}
V_{t}=V_{0}\exp\left\{\alpha\mu t-\frac{1}{2}\alpha^{2}\int_{0}^{t}\nu_{s}ds+(1-\alpha)rt+\alpha\int_{0}^{t}\sqrt{\nu_{s}}dB_{s}\right\}.
\end{equation}
Hence, we get
\begin{equation}
E[u(V_{t})]=\frac{1}{\theta}V_{0}^{\theta}
e^{\theta(\alpha\mu+(1-\alpha)r)t}
E\left[e^{\theta(\alpha\int_{0}^{t}\sqrt{\nu_{s}}dB_{s}
-\frac{1}{2}\alpha^{2}\int_{0}^{t}\nu_{s}ds)}\right].
\end{equation}

\begin{lemma}\label{HLemma}
For any $\nu>0$,
\begin{align}
&\lim_{t\rightarrow\infty}\frac{1}{t}\log E_{\nu_{0}=\nu}
\left[e^{\theta(\alpha\int_{0}^{t}\sqrt{\nu_{s}}dB_{s}-\frac{1}{2}\alpha^{2}\int_{0}^{t}\nu_{s}ds)}\right]
\\
&=\frac{\kappa^{2}\gamma}{\delta^{2}}-\frac{\kappa\gamma}{\delta^{2}}
\sqrt{\kappa^{2}-\delta^{2}\theta^{2}\alpha^{2}(1-\rho^{2})+\delta^{2}\theta\alpha^{2}
-2\delta\kappa\alpha\rho\theta}
-\frac{\theta\alpha\rho\kappa}{\delta}\gamma
.\nonumber
\end{align}
\end{lemma}

\begin{proof}
Write $B_{t}=\rho W_{t}+\sqrt{1-\rho^{2}}Z_{t}$, where $Z_{t}$ is a standard Brownian motion independent of
$W_{t}$. Let $\mathcal{F}_{t}^{\nu}:=\sigma(\nu_{s},0\leq s\leq t)$ be the natural sigma field of the volatility process up to time $t$.
It is easy to compute that
\begin{align}
&E_{\nu_{0}=\nu}\left[e^{\theta(\alpha\int_{0}^{t}\sqrt{\nu_{s}}dB_{s}-\frac{1}{2}\alpha^{2}\int_{0}^{t}\nu_{s}ds)}\right]
\\
&=E_{\nu_{0}=\nu}\left[e^{\theta(\alpha\int_{0}^{t}\sqrt{\nu_{s}}\rho dW_{s}+\alpha\int_{0}^{t}\sqrt{\nu_{s}}\sqrt{1-\rho^{2}}dZ_{s}
-\frac{1}{2}\alpha^{2}\int_{0}^{t}\nu_{s}ds)}\right]\nonumber
\\
&=E_{\nu_{0}=\nu}\left[E\left[e^{\theta(\alpha\int_{0}^{t}\sqrt{\nu_{s}}\rho dW_{s}+\alpha\int_{0}^{t}\sqrt{\nu_{s}}\sqrt{1-\rho^{2}}dZ_{s}
-\frac{1}{2}\alpha^{2}\int_{0}^{t}\nu_{s}ds)}\bigg|\mathcal{F}_{t}^{\nu}\right]\right]\nonumber
\\
&=E_{\nu_{0}=\nu}\left[e^{\theta\alpha\int_{0}^{t}\sqrt{\nu_{s}}\rho dW_{s}+\frac{1}{2}\theta^{2}\alpha^{2}(1-\rho^{2})\int_{0}^{t}\nu_{s}ds
-\frac{1}{2}\theta\alpha^{2}\int_{0}^{t}\nu_{s}ds}\right]\nonumber
\\
&=E_{\nu_{0}=\nu}\left[e^{\frac{\theta\alpha\rho}{\delta}\nu_{t}+(\frac{\theta^{2}\alpha^{2}(1-\rho^{2})-\theta\alpha^{2}}{2}
+\frac{\kappa\alpha\rho}{\delta}\theta)\int_{0}^{t}\nu_{s}ds}\right]
e^{-\frac{\theta\alpha\rho}{\delta}\nu-\frac{\theta\alpha\rho\kappa}{\delta}\gamma t}
,\nonumber
\end{align}
where the last step was due to the fact that $\nu_{t}-\nu_{0}=\int_{0}^{t}\kappa(\gamma-\nu_{s})ds+\int_{0}^{t}\delta\sqrt{\nu_{s}}dW_{s}$.

Let $u(t,\nu):=E_{\nu_{0}=\nu}\left[e^{\frac{\theta\alpha\rho}{\delta}\nu_{t}+(\frac{\theta^{2}\alpha^{2}(1-\rho^{2})-\theta\alpha^{2}}{2}
+\frac{\kappa\alpha\rho}{\delta}\theta)\int_{0}^{t}\nu_{s}ds}\right]$. Feynman-Kac formula implies
that $u(t,\nu)$ satisfies the following partial differential equation,
\begin{equation}
\begin{cases}
\frac{\partial u}{\partial t}=\kappa(\gamma-\nu)\frac{\partial u}{\partial\nu}
+\frac{1}{2}\delta^{2}\nu\frac{\partial^{2}u}{\partial\nu^{2}}
+\left(\frac{\theta^{2}\alpha^{2}(1-\rho^{2})-\theta\alpha^{2}}{2}
+\frac{\kappa\alpha\rho}{\delta}\theta\right)\nu u,
\\
u(0,\nu)=e^{\frac{\theta\alpha\rho}{\delta}\nu}.
\end{cases}
\end{equation}
Let us try $u(t,\nu)=e^{A(t)+B(t)\nu}$ and it is easy to see that $A(t),B(t)$ satisfy the following system of ordinary differential equations,
\begin{equation}
\begin{cases}
A'(t)=\kappa\gamma B(t),
\\
B'(t)=-\kappa B(t)+\frac{1}{2}\delta^{2}B(t)^{2}+\left(\frac{\theta^{2}\alpha^{2}(1-\rho^{2})-\theta\alpha^{2}}{2}
+\frac{\kappa\alpha\rho}{\delta}\theta\right),
\\
A(0)=0,\quad B(0)=\frac{\theta\alpha\rho}{\delta}.
\end{cases}
\end{equation}

We claim that there are two distinct solutions to the quadratic equation
\begin{equation}\label{quadratic}
\frac{1}{2}\delta^{2}x^{2}-\kappa x+\left(\frac{\theta^{2}\alpha^{2}(1-\rho^{2})-\theta\alpha^{2}}{2}
+\frac{\kappa\alpha\rho}{\delta}\theta\right)=0,
\end{equation}
and $B(t)$ converges to the smaller solution of \eqref{quadratic}. 

We can compute that 
\begin{equation}
\Delta:=\kappa^{2}-2\delta^{2}\left(\frac{\theta^{2}\alpha^{2}(1-\rho^{2})-\theta\alpha^{2}}{2}
+\frac{\kappa\alpha\rho}{\delta}\theta\right).
\end{equation}
If $\alpha=0$, then $\Delta=\kappa^{2}>0$. If $\alpha\neq 0$, then,
\begin{align}
\Delta&=(\kappa^{2}+\delta^{2}\theta^{2}\alpha^{2}\rho^{2}-2\delta\kappa\alpha\rho\theta)+\delta^{2}\alpha^{2}(\theta-\theta^{2})
\\
&=(\kappa-\delta\theta\alpha\rho)^{2}+\delta^{2}\alpha^{2}(\theta-\theta^{2})>0,\nonumber
\end{align}
since $\theta\in(0,1)$. Hence \eqref{quadratic} has two distinct solutions.
$B'(t)$ is positive when $B(t)$ is smaller than the smaller solution of \eqref{quadratic} or larger
than the larger solution of $\eqref{quadratic}$. $B'(t)$ is negative if $B(t)$ lies between the two solutions
of \eqref{quadratic}. Therefore, $B(t)$ converges to the smaller solution of \eqref{quadratic} if
$B(0)=\frac{\theta\alpha\rho}{\delta}$
is less than the larger solution of \eqref{quadratic}. When $\alpha=0$, $B(0)=0$ and the large solution
of \eqref{quadratic} equals to $\frac{2\kappa}{\delta^{2}}>0$. Hence, we can assume that $\alpha>0$.
Let
\begin{equation}
H(x):=\frac{1}{2}\delta^{2}x^{2}-\kappa x+\left(\frac{\theta^{2}\alpha^{2}(1-\rho^{2})-\theta\alpha^{2}}{2}
+\frac{\kappa\alpha\rho}{\delta}\theta\right).
\end{equation}
It is easy to check that
\begin{equation}
H\left(\frac{\theta\alpha\rho}{\delta}\right)
=\frac{(\theta^{2}-\theta)\alpha^{2}}{2}<0,
\end{equation}
since $\theta\in(0,1)$ and $\alpha>0$. Therefore, we conclude that $B(0)$ is less than
the larger solution of \eqref{quadratic} and
\begin{equation}
B(t)\rightarrow\frac{\kappa}{\delta^{2}}-\frac{1}{\delta^{2}}\sqrt{\kappa^{2}-2\delta^{2}\left(\frac{\theta^{2}\alpha^{2}(1-\rho^{2})-\theta\alpha^{2}}{2}
+\frac{\kappa\alpha\rho}{\delta}\theta\right)},
\end{equation}
as $t\rightarrow\infty$ and hence
\begin{equation}
\frac{A(t)}{t}=\frac{1}{t}\kappa\gamma\int_{0}^{t}B(s)ds
\rightarrow\frac{\kappa^{2}\gamma}{\delta^{2}}-\frac{\kappa\gamma}{\delta^{2}}
\sqrt{\kappa^{2}-\delta^{2}\theta^{2}\alpha^{2}(1-\rho^{2})+\delta^{2}\theta\alpha^{2}
-2\delta\kappa\alpha\rho\theta},
\end{equation}
as $t\rightarrow\infty$. 
Recall that 
$E_{\nu_{0}=\nu}\left[e^{\theta(\alpha\int_{0}^{t}\sqrt{\nu_{s}}dB_{s}-\frac{1}{2}\alpha^{2}\int_{0}^{t}\nu_{s}ds)}\right]
=u(t,\nu)e^{-\frac{\theta\alpha\rho}{\delta}\nu-\frac{\theta\alpha\rho\kappa}{\delta}\gamma t}$.
Hence, we conclude that
\begin{align}
&\lim_{t\rightarrow\infty}\frac{1}{t}\log E_{\nu_{0}=\nu}
\left[e^{\theta(\alpha\int_{0}^{t}\sqrt{\nu_{s}}dB_{s}-\frac{1}{2}\alpha^{2}\int_{0}^{t}\nu_{s}ds)}\right]
\\
&=\frac{\kappa^{2}\gamma}{\delta^{2}}-\frac{\kappa\gamma}{\delta^{2}}
\sqrt{\kappa^{2}-\delta^{2}\theta^{2}\alpha^{2}(1-\rho^{2})+\delta^{2}\theta\alpha^{2}
-2\delta\kappa\alpha\rho\theta}
-\frac{\theta\alpha\rho\kappa}{\delta}\gamma
.\nonumber
\end{align}
\end{proof}

\begin{theorem}
\begin{align}
\Lambda(\alpha)&=\lim_{t\rightarrow\infty}\frac{1}{t}\log E[u(V_{t})]
\\
&=\frac{\kappa^{2}\gamma}{\delta^{2}}-\frac{\kappa\gamma}{\delta^{2}}
\sqrt{\kappa^{2}-\delta^{2}\theta^{2}\alpha^{2}(1-\rho^{2})+\delta^{2}\theta\alpha^{2}
-2\delta\kappa\alpha\rho\theta}
\nonumber
\\
&\qquad\qquad\qquad
-\frac{\theta\alpha\rho\kappa}{\delta}\gamma+\theta\alpha\mu+\theta(1-\alpha)r.
\nonumber
\end{align}
Let us define
\begin{equation}\label{DefnC}
\begin{cases}
C_{0}:=\frac{\kappa^{2}\gamma}{\delta^{2}}+\theta r,
\\
C_{1}:=\frac{\kappa^{2}\gamma^{2}}{\delta^{4}}(\delta^{2}\theta-\delta^{2}\theta^{2}(1-\rho^{2})),
\\
C_{2}:=\delta\kappa\rho\theta\cdot\frac{\kappa^{2}\gamma^{2}}{\delta^{4}},
\\
C_{3}:=\frac{\kappa^{4}\gamma^{2}}{\delta^{4}},
\\
C_{4}:=-\frac{\theta\rho\kappa}{\delta}\gamma+\theta(\mu-r).
\end{cases}
\end{equation}
When $C_{4}+\frac{C_{2}}{\sqrt{C_{3}}}\leq 0$, the optimal $\alpha^{\ast}=0$
and when $C_{4}\geq\sqrt{C_{1}}$, the optimal $\alpha^{\ast}=1$.
Finally, if $-\frac{C_{2}}{\sqrt{C_{3}}}<C_{4}<\sqrt{C_{1}}$,
\begin{equation}
\alpha^{\ast}=
\begin{cases}
\alpha^{\dagger} &\text{if $\alpha^{\dagger}<1$},
\\
1 &\text{otherwise},
\end{cases}
\end{equation}
where
\begin{equation}
\alpha^{\dagger}=\frac{1}{C_{1}}\left[C_{2}
+C_{4}\sqrt{\frac{C_{1}C_{3}-C_{2}^{2}}{C_{1}-C_{4}^{2}}}\right].
\end{equation}
\end{theorem}

\begin{proof}
Recall that $E[u(V_{t})]=\frac{1}{\theta}V_{0}^{\theta}
e^{\theta(\alpha\mu+(1-\alpha)r)t}
E\left[e^{\theta(\alpha\int_{0}^{t}\sqrt{\nu_{s}}dB_{s}
-\frac{1}{2}\alpha^{2}\int_{0}^{t}\nu_{s}ds)}\right]$.
By Lemma \ref{HLemma}, we have
\begin{align}
\Lambda(\alpha)&=\lim_{t\rightarrow\infty}\frac{1}{t}\log E[u(V_{t})]
\\
&=\frac{\kappa^{2}\gamma}{\delta^{2}}-\frac{\kappa\gamma}{\delta^{2}}
\sqrt{\kappa^{2}-\delta^{2}\theta^{2}\alpha^{2}(1-\rho^{2})+\delta^{2}\theta\alpha^{2}-2\delta\kappa\alpha\rho\theta}
\nonumber
\\
&\qquad\qquad\qquad
-\frac{\theta\alpha\rho\kappa}{\delta}\gamma+\theta\alpha\mu+\theta(1-\alpha)r
\nonumber
\\
&=-\sqrt{C_{1}\alpha^{2}-2C_{2}\alpha+C_{3}}+C_{4}\alpha+C_{0}.\nonumber
\end{align}
From the definition in \eqref{DefnC},
it is clear that $C_{0},C_{1},C_{2},C_{3}>0$. But $C_{4}$ may or may not
be positive. 

It is easy to compute that
\begin{equation}
\Lambda'(\alpha)
=C_{4}-\frac{C_{1}\alpha-C_{2}}{\sqrt{C_{1}\alpha^{2}-2C_{2}\alpha+C_{3}}},
\end{equation}
and
\begin{equation}
\Lambda''(\alpha)
=-\frac{C_{1}\sqrt{C_{1}}(C_{1}C_{3}-C_{2}^{2})}
{((C_{1}\alpha-C_{2})^{2}+(C_{1}C_{3}-C_{2}^{2}))^{3/2}}.
\end{equation}
On the other hand, since $\theta\in(0,1)$, we have
\begin{align}
C_{1}C_{3}-C_{2}^{2}
&=\frac{\kappa^{2}\gamma^{2}}{\delta^{4}}\frac{\kappa^{4}\gamma^{2}}{\delta^{4}}
(\delta^{2}\theta-\delta^{2}\theta^{2}(1-\rho^{2}))
-\delta^{2}\kappa^{2}\rho^{2}\theta^{2}\frac{\kappa^{4}\gamma^{4}}{\delta^{8}}
\\
&=\frac{\kappa^{6}\gamma^{4}}{\delta^{8}}(\delta^{2}\theta-\delta^{2}\theta^{2}(1-\rho^{2})-\delta^{2}\rho^{2}\theta^{2})
\nonumber
\\
&=\frac{\kappa^{6}\gamma^{4}}{\delta^{6}}(\theta-\theta^{2})>0.\nonumber
\end{align}
Hence, we conclude that $\Lambda''(\alpha)<0$
for any $\alpha$, i.e. $\Lambda(\alpha)$ is strictly concave in $\alpha$.

Note that $\Lambda'(0)=C_{4}+\frac{C_{2}}{\sqrt{C_{3}}}$. If $C_{4}+\frac{C_{2}}{\sqrt{C_{3}}}\leq 0$, since
$\Lambda(\alpha)$ is strictly concave, the maximum must be achieved
at $\alpha^{\ast}=0$.
Now assume that $C_{4}+\frac{C_{2}}{\sqrt{C_{3}}}>0$.
When $C_{4}>\sqrt{C_{1}}$, it is easy
to check that $\Lambda'(\alpha)\sim(C_{4}-\sqrt{C_{1}})$ as $\alpha\rightarrow\infty$,
and since $\Lambda(\alpha)$ is strictly concave, it yields that 
$\Lambda(\alpha)$
is increasing in $\alpha\geq 0$ and the maximum is achieved
at $\alpha^{\dagger}=1$.
If $C_{4}=\sqrt{C_{1}}$,
\begin{align}
\Lambda'(\alpha)
&=\frac{C_{4}\sqrt{C_{1}\alpha^{2}-2C_{2}\alpha+C_{3}}-C_{1}\alpha+C_{2}}
{\sqrt{C_{1}\alpha^{2}-2C_{2}\alpha+C_{3}}}
\\
&=\frac{\frac{C_{4}}{C_{1}}
\sqrt{(C_{1}\alpha-C_{2})^{2}+(C_{3}C_{1}-C_{2}^{2})}-C_{1}\alpha+C_{2}}
{\sqrt{C_{1}\alpha^{2}-2C_{2}\alpha+C_{3}}}
\nonumber
\\
&=\frac{
\sqrt{(C_{1}\alpha-C_{2})^{2}+(C_{3}C_{1}-C_{2}^{2})}-(C_{1}\alpha-C_{2})}
{\sqrt{C_{1}\alpha^{2}-2C_{2}\alpha+C_{3}}}>0,
\nonumber
\end{align}
since $C_{3}C_{1}-C_{2}^{2}>0$. Thus, $\alpha^{\ast}=1$
when $C_{4}=\sqrt{C_{1}}$.

Now assume that $-\frac{C_{2}}{\sqrt{C_{3}}}<C_{4}<\sqrt{C_{1}}$. Then $\Lambda'(0)=C_{4}+\frac{C_{2}}{\sqrt{C_{3}}}>0$
and $\Lambda(\alpha)\rightarrow-\infty$ as $\alpha\rightarrow\infty$.
Thus, there exists a unique global maximum on $(0,\infty)$, given
by $\alpha^{\dagger}$. So that
\begin{equation}
\Lambda'(\alpha^{\dagger})=\sqrt{C_{1}}
\left[\frac{C_{4}}{\sqrt{C_{1}}}-\frac{(C_{1}\alpha^{\dagger}-C_{2})}
{\sqrt{(C_{1}\alpha^{\dagger}-C_{2})^{2}+C_{1}C_{3}-C_{2}^{2}}}\right]=0.
\end{equation}
$C_{1}\alpha^{\dagger}-C_{2}$ has the same sign as $C_{4}$ which is positive.
Hence, we can solve for $\alpha^{\dagger}$ and get
\begin{equation}
\alpha^{\dagger}=\frac{1}{C_{1}}\left[C_{2}
+C_{4}\sqrt{\frac{C_{1}C_{3}-C_{2}^{2}}{C_{1}-C_{4}^{2}}}\right].
\end{equation}
The optimal $\alpha^{\ast}$ is given by
\begin{equation}
\alpha^{\ast}=
\begin{cases}
\alpha^{\dagger} &\text{if $\alpha^{\dagger}<1$},
\\
1 &\text{otherwise}.
\end{cases}
\end{equation}
\end{proof}

\section{3/2 Model}\label{ThreeHalvesSection}

Let us assume that the stock price follows a 3/2 model, namely, 
\begin{equation}
\begin{cases}
dS_{t}=\mu S_{t}dt+\sqrt{\nu_{t}}S_{t}dB_{t},
\\
d\nu_{t}=\kappa\nu_{t}(\gamma-\nu_{t})+\delta\nu_{t}^{3/2}dW_{t},
\end{cases}
\end{equation}
where $B_{t}$ and $W_{t}$ are two standard Brownian motions, which
are assumed to be independent for simplicity.

Therefore, the wealth process satisfies
\begin{equation}
\begin{cases}
dV_{t}=\alpha\mu V_{t}dt+\alpha\sqrt{\nu_{t}}V_{t}dB_{t}+(1-\alpha)rV_{t}dt,
\\
d\nu_{t}=\kappa(\gamma-\nu_{t})dt+\delta\nu_{t}^{3/2}dW_{t},
\end{cases}
\end{equation}
Then, we have
\begin{equation}
V_{t}=V_{0}\exp
\left\{\alpha\mu t-\frac{1}{2}\alpha^{2}\int_{0}^{t}\nu_{s}ds
+(1-\alpha)rt+\alpha\int_{0}^{t}\sqrt{\nu_{s}}dB_{s}\right\}.
\end{equation}
Hence, we get
\begin{align}
E[u(V_{t})]&=\frac{1}{\theta}V_{0}^{\theta}
e^{\theta(\alpha\mu+(1-\alpha)r)t}
E\left[e^{\theta(\alpha\int_{0}^{t}\sqrt{\nu_{s}}dB_{s}
-\frac{1}{2}\alpha^{2}\int_{0}^{t}\nu_{s}ds)}\right]
\\
&=\frac{1}{\theta}V_{0}^{\theta}
e^{\theta(\alpha\mu+(1-\alpha)r)t}
E\left[e^{-\frac{1}{2}\alpha^{2}(\theta-\theta^{2})\int_{0}^{t}\nu_{s}ds}\right].
\nonumber
\end{align}
The volatility process $\nu_{t}$ is not an affine process but it is still analytically tractable.
The Laplace tranform of $\int_{0}^{t}\nu_{s}ds$ is known, see e.g. 
Lewis \cite{Lewis}. 
\begin{equation}
E_{\nu_{0}}\left[e^{-\lambda\int_{0}^{t}\nu_{s}ds}\right]
=\frac{\Gamma(b-a)}{\Gamma(b)}
\left(\frac{2\kappa\gamma}{\delta^{2}\nu_{0}(e^{\kappa\gamma t}-1)}\right)^{a}
M\left(a,b,-\frac{2\kappa\gamma}{\delta^{2}\nu_{0}(e^{\kappa\gamma t}-1)}\right),
\end{equation}
where
\begin{equation}
\begin{cases}
a:=-\left(\frac{1}{2}+\frac{\kappa}{\delta^{2}}\right)
+\sqrt{\left(\frac{1}{2}+\frac{\kappa}{\delta^{2}}\right)^{2}
+\frac{2\lambda}{\delta^{2}}},
\\
b:=1+2\sqrt{\left(\frac{1}{2}+\frac{\kappa}{\delta^{2}}\right)^{2}
+\frac{2\lambda}{\delta^{2}}},
\end{cases}
\end{equation}
and $\Gamma(\cdot)$ is the standard Gamma function and
$M(a,b,z):=\sum_{n=0}^{\infty}\frac{(a)_{n}}{(b)_{n}}\frac{z^{n}}{n!}$
is the confluent hypergeometric function, also known 
as Kummer's function (see e.g. Abramowitz and Stegun \cite{Abramowitz}), where $(c)_{0}:=1$
and $(c)_{n}:=c(c+1)\cdots(c+n-1)$ for $n\geq 1$.
In our case,
\begin{equation}
\lambda:=\frac{1}{2}\alpha^{2}(\theta-\theta^{2})>0,
\end{equation}
since $\theta\in(0,1)$. We are interested in the asymptotic behavior
of the Laplace transform as $t\rightarrow\infty$.
As $t\rightarrow\infty$, 
$-\frac{2\kappa\gamma}{\delta^{2}\nu_{0}(e^{\kappa\gamma t}-1)}
\rightarrow 0$ and $M(a,b,-\frac{2\kappa\gamma}{\delta^{2}\nu_{0}(e^{\kappa\gamma t}-1)})
\rightarrow 1$. Then, it is easy to see that
\begin{align}
\lim_{t\rightarrow\infty}\frac{1}{t}\log
E\left[e^{-\frac{1}{2}\alpha^{2}(\theta-\theta^{2})\int_{0}^{t}\nu_{s}ds}\right]
&=-a\kappa\gamma
\\
&=\kappa\gamma\left(\frac{1}{2}+\frac{\kappa}{\delta^{2}}\right)
-\kappa\gamma\sqrt{\left(\frac{1}{2}+\frac{\kappa}{\delta^{2}}\right)
+\frac{\alpha^{2}(\theta-\theta^{2})}{\delta^{2}}}.
\nonumber
\end{align}
Hence, we conclude that
\begin{align}
\Lambda(\alpha)&=\lim_{t\rightarrow\infty}\frac{1}{t}\log E[u(V_{t})]
\\
&=\theta\alpha\mu+\theta(1-\alpha)r+\kappa\gamma\left(\frac{1}{2}+\frac{\kappa}{\delta^{2}}\right)
-\kappa\gamma\sqrt{\left(\frac{1}{2}+\frac{\kappa}{\delta^{2}}\right)
+\frac{\alpha^{2}(\theta-\theta^{2})}{\delta^{2}}}.
\nonumber
\end{align}
It is straightforward to check that $\Lambda''(\alpha)<0$
and it is easy to compute that
\begin{equation}\label{twosolutions}
\Lambda'(\alpha)=\theta(\mu-r)-\kappa\gamma\frac{(\theta-\theta^{2})}{\delta^{2}}
\frac{\alpha}{\sqrt{\left(\frac{1}{2}+\frac{\kappa}{\delta^{2}}\right)
+\frac{\alpha^{2}(\theta-\theta^{2})}{\delta^{2}}}}.
\end{equation}
When $\mu-r\leq 0$, since $\Lambda(\alpha)$ is strictly concave,
$\Lambda(\alpha)$ is decreasing for $\alpha\geq 0$
and thus the optimal $\alpha^{\ast}$ is achieved at $\alpha^{\ast}=0$.
Now, assume that $\mu-r>0$.
When $\theta(\mu-r)-\frac{\kappa\gamma}{\delta}\sqrt{\theta-\theta^{2}}\geq 0$,
$\Lambda'(\alpha)\geq 0$ for any $\alpha$ and the optimal $\alpha^{\ast}$
is achieved at $\alpha^{\ast}=1$.
When $\theta(\mu-r)-\frac{\kappa\gamma}{\delta}\sqrt{\theta-\theta^{2}}
<0$,
there exists a unique global maximum of $\Lambda(\alpha)$
achieved at $\alpha^{\dagger}\in(0,\infty)$ 
so that $\Lambda'(\alpha^{\dagger})=0$.
Observe that if $\Lambda'(\alpha)=0$ in \eqref{twosolutions}, then $\alpha$ has the same
sign as $\mu-r>0$.
After some algebraic manipulations, we get
\begin{equation}
\alpha^{\dagger}
=\frac{\theta(\mu-r)\delta^{2}}{\sqrt{\kappa^{2}\gamma^{2}(\theta-\theta^{2})
-\theta^{2}(\mu-r)^{2}\delta^{2}}\sqrt{\theta-\theta^{2}}}
\sqrt{\frac{1}{2}+\frac{\kappa}{\delta^{2}}}.
\end{equation}
Thus, $\alpha^{\ast}=\alpha^{\dagger}$ if $\alpha^{\dagger}<1$
and $\alpha^{\ast}=1$ otherwise.

We summarise our results in the following theorem.
\begin{theorem}
\begin{align}
\Lambda(\alpha)&=\lim_{t\rightarrow\infty}\frac{1}{t}\log E[u(V_{t})]
\\
&=\theta\alpha\mu+\theta(1-\alpha)r+\kappa\gamma\left(\frac{1}{2}+\frac{\kappa}{\delta^{2}}\right)
-\kappa\gamma\sqrt{\left(\frac{1}{2}+\frac{\kappa}{\delta^{2}}\right)
+\frac{\alpha^{2}(\theta-\theta^{2})}{\delta^{2}}}.
\nonumber
\end{align}
The optimal $\alpha^{\ast}$ is given by $\alpha^{\ast}=0$
if $\mu-r\leq 0$ and $\alpha^{\ast}=1$ 
if $\theta(\mu-r)-\frac{\kappa\gamma}{\delta}\sqrt{\theta-\theta^{2}}\geq 0$ 
and if $\mu>r$ 
and $\theta(\mu-r)-\frac{\kappa\gamma}{\delta}\sqrt{\theta-\theta^{2}}<0$,
the optimal $\alpha^{\ast}$ is given by
\begin{equation}
\alpha^{\ast}=
\begin{cases}
\alpha^{\dagger} &\text{if $\alpha^{\dagger}<1$},
\\
1 &\text{otherwise},
\end{cases}
\end{equation}
where
\begin{equation}
\alpha^{\dagger}
=\frac{\theta(\mu-r)\delta^{2}}{\sqrt{\kappa^{2}\gamma^{2}(\theta-\theta^{2})
-\theta^{2}(\mu-r)^{2}\delta^{2}}\sqrt{\theta-\theta^{2}}}
\sqrt{\frac{1}{2}+\frac{\kappa}{\delta^{2}}}.
\end{equation}
\end{theorem}

\begin{remark}
For simplicity, we only considered the case when $\rho=0$. Indeed, for general $-1\leq\rho\leq 1$, the joint
Fourier-Laplace transform of the logarithm of the spot price, i.e. $\log(S_{t}/S_{0})$ and the total integrated
variance, i.e. $\int_{0}^{t}\nu_{s}ds$ is also known in the close-form, see e.g. Carr and Sun \cite{Carr} and our methods
can still be applied to obtain the optimal strategy $\alpha^{\ast}$. But the computations would be more involved.
\end{remark}

\section{Jump Diffusion Model}\label{JSection}

Let us assume that the stock price follows a jump diffusion model. 
More precisely,
\begin{equation}
dS_{t}=\mu S_{t-}dt+\sigma S_{t-}dB_{t}+S_{t-}dJ_{t},
\end{equation}
where $J_{t}=\sum_{i=1}^{N_{t}}(Y_{i}-1)$, where $Y_{i}$ are i.i.d. random variables distributed on $(0,\infty)$ with a smooth and bounded
probability density function
and $Y_{i}$ are independent of $N_{t}$ which is a standard Poisson process with intensity $\lambda>0$.
We further assume that $E[Y_{1}]<\infty$.

The wealth process satisfies
\begin{equation}
dV_{t}=\alpha\mu V_{t-}dt+\alpha\sigma V_{t-}dB_{t}+\alpha V_{t-}dJ_{t}+(1-\alpha)rV_{t-}dt.
\end{equation}
Then, we have
\begin{equation}
V_{t}=V_{0}e^{\alpha\mu t+(1-\alpha)rt+\alpha\sigma B_{t}-\frac{1}{2}\alpha^{2}\sigma^{2}t+\sum_{i=1}^{N_{t}}\log(\alpha(Y_{i}-1)+1)}.
\end{equation}
Therefore,
\begin{equation}
\Lambda(\alpha)=\theta[\alpha\mu+(1-\alpha)r]+\frac{1}{2}(\theta^{2}-\theta)\alpha^{2}\sigma^{2}
+\lambda(E[(\alpha(Y_{1}-1)+1)^{\theta}]-1).
\end{equation}

\begin{remark}
(i) If $Y_{1}\equiv y$ is a positive constant, then
\begin{equation}
\Lambda(\alpha)=\theta[\alpha\mu+(1-\alpha)r]
+\frac{1}{2}(\theta^{2}-\theta)\alpha^{2}\sigma^{2}
+\lambda((\alpha(y-1)+1)^{\theta}-1).
\end{equation}

(ii) If $Y_{1}$ is exponentially distributed with parameter $\rho>0$, then
\begin{align}
E[(\alpha(Y_{1}-1)+1)^{\theta}]
&=\int_{0}^{\infty}(\alpha y+1-\alpha)^{\theta}e^{-\rho y}\rho dy
\\
&=e^{\rho(\frac{1}{\alpha}-1)}\left(\alpha/\rho\right)^{\theta}
\Gamma\left(\theta+1,\rho\left(\frac{1}{\alpha}-1\right)\right),\nonumber
\end{align}
where $\Gamma(s,x):=\int_{x}^{\infty}t^{s-1}e^{-t}dt$ is an upper
incomplete Gamma function.
\end{remark}

It is easy to compute that
\begin{equation}
\Lambda'(\alpha)
=\theta(\mu-r)+(\theta^{2}-\theta)\sigma^{2}\alpha+\lambda\theta E\left[(\alpha(Y_{1}-1)+1)^{\theta-1}(Y_{1}-1)\right],
\end{equation}
and
\begin{equation}
\Lambda''(\alpha)
=(\theta^{2}-\theta)\sigma^{2}+\lambda\theta(\theta-1)E\left[(\alpha(Y_{1}-1)+1)^{\theta-2}(Y_{1}-1)^{2}\right]<0,
\end{equation}
since $\theta\in(0,1)$ and $\alpha\in[0,1]$. 

In the expression of $\Lambda(\alpha)$, 
\begin{align}
|E[(\alpha(Y_{1}-1)+1)^{\theta}]|
&\leq E[(\alpha|Y_{1}-1|+1)^{\theta}]
\\
&\leq E[(\alpha(Y_{1}+1)+1)^{\theta}]\nonumber
\\
&\leq E[(\alpha(Y_{1}+1)+1)]\nonumber
\\
&=\alpha(E[Y_{1}]+1)+1,\nonumber
\end{align}
since $\theta\in(0,1)$ and $\alpha(Y_{1}+1)+1\geq 1$ a.s. The coefficient
of $\alpha^{2}$ term in $\Lambda(\alpha)$ is $\frac{1}{2}(\theta^{2}-\theta)\sigma^{2}$
which is negative. Thus, $\Lambda(\alpha)\rightarrow-\infty$
as $\alpha\rightarrow\infty$. Recall that $\Lambda(\alpha)$
is strictly concave. 
Therefore, if $\Lambda'(0)=\theta(\mu-r)+\lambda\theta(E[Y_{1}]-1)\leq 0$,
then, the optimal $\alpha^{\ast}$ is achieved at $\alpha^{\ast}=0$.
If $\Lambda'(0)=\theta(\mu-r)+\lambda\theta(E[Y_{1}]-1)>0$,
then, there exists a unique $\alpha^{\dagger}\in(0,\infty)$
so that $\Lambda'(\alpha^{\dagger})=0$. 
In this case, the optimal $\alpha^{\ast}$ is given by
\begin{equation}
\alpha^{\ast}
=
\begin{cases}
1 &\text{if $\alpha^{\dagger}\geq 1$},
\\
\alpha^{\dagger} &\text{if $\alpha^{\dagger}\in(0,1)$}.
\end{cases}
\end{equation}

We summarize our conclusions in the following theorem.
\begin{theorem}
\begin{equation}
\Lambda(\alpha)=\theta[\alpha\mu+(1-\alpha)r]+\frac{1}{2}(\theta^{2}-\theta)\alpha^{2}\sigma^{2}
+\lambda(E[(\alpha(Y_{1}-1)+1)^{\theta}]-1).
\end{equation}
When $\theta(\mu-r)+\lambda\theta(E[Y_{1}]-1)\leq 0$,
$\alpha^{\ast}=0$. Otherwise,
\begin{equation}
\alpha^{\ast}
=
\begin{cases}
1 &\text{if $\alpha^{\dagger}\geq 1$},
\\
\alpha^{\dagger} &\text{if $\alpha^{\dagger}\in(0,1)$},
\end{cases}
\end{equation}
where $\alpha^{\dagger}$ is the unique positive solution to
\begin{equation}
\theta(\mu-r)+(\theta^{2}-\theta)\sigma^{2}\alpha
+\lambda\theta E\left[(\alpha(Y_{1}-1)+1)^{\theta-1}(Y_{1}-1)\right]=0.
\end{equation}
\end{theorem}

\section{Black-Scholes Model with Vasicek Interest Rate}\label{VSection}

Let us assume that the stock price follows a Black-Scholes model 
with constant drift $\mu$ and volatility $\sigma$ and the interest rate $r_{t}$ follows a Vasicek model.
The Vasicek model is a standard interest rate model, introduced by Vasicek \cite{Vasicek}.
The wealth process satisfies the following stochastic differential equation.
\begin{equation}
dV_{t}=\alpha\mu V_{t}dt+\alpha\sigma V_{t}dB_{t}+(1-\alpha)r_{t}V_{t}dt,
\end{equation}
where $B_{t}$ is a standard Brownian motion starting at $0$ at time $0$ and
\begin{equation}
dr_{t}=\kappa(\gamma-r_{t})dt+\delta dW_{t},
\end{equation}
where $W_{t}$ is a standard Brownian motion so that $\langle W,B\rangle_{t}=\rho t$, where $-1\leq\rho\leq 1$ is the correlation.

Therefore, the wealth process is given by
\begin{equation}
V_{t}=V_{0}e^{\alpha\mu t+\alpha\sigma B_{t}-\frac{1}{2}\alpha^{2}\sigma^{2}t+(1-\alpha)\int_{0}^{t}r_{s}ds}.
\end{equation}

\begin{lemma}\label{VLemma}
For any $r_{0}=r>0$,
\begin{align}
&\lim_{t\rightarrow\infty}\frac{1}{t}\log
E_{r_{0}=r}
\left[e^{\theta(\alpha\sigma B_{t}+(1-\alpha)\int_{0}^{t}r_{s}ds)}\right]
\\
&=\kappa\gamma\left(\frac{\theta(1-\alpha)}{\kappa}
+\frac{\theta\alpha\sigma\rho}{\delta}\right)
+\frac{\delta^{2}}{2}\left(\frac{\theta(1-\alpha)}{\kappa}
+\frac{\theta\alpha\sigma\rho}{\delta}\right)^{2}
\nonumber
\\
&\qquad\qquad\qquad
-\frac{\theta\alpha\sigma\kappa\gamma\rho}{\delta}
+\frac{1}{2}\theta^{2}\alpha^{2}\sigma^{2}(1-\rho^{2})
.\nonumber
\end{align}
\end{lemma}

\begin{proof}
Write $B_{t}=\rho W_{t}+\sqrt{1-\rho^{2}}Z_{t}$, where $Z_{t}$ is a standard Brownian motion independent of $B_{t}$ and $W_{t}$.
Therefore, we have
\begin{align}
&E_{r_{0}=r}\left[e^{\theta(\alpha\sigma B_{t}+(1-\alpha)\int_{0}^{t}r_{s}ds)}\right]
\\
&=E_{r_{0}=r}\left[e^{\theta(\alpha\sigma\rho W_{t}+\alpha\sigma\sqrt{1-\rho^{2}}Z_{t}+(1-\alpha)\int_{0}^{t}r_{s}ds)}\right]
\nonumber
\\
&=E_{r_{0}=r}\left[e^{\frac{1}{2}\theta^{2}\alpha^{2}\sigma^{2}(1-\rho^{2})t+\theta(\alpha\sigma\rho W_{t}+(1-\alpha)\int_{0}^{t}r_{s}ds)}\right]
\nonumber
\\
&=E_{r_{0}=r}\left[e^{
\theta(1-\alpha)\int_{0}^{t}r_{s}ds+\frac{\theta\alpha\sigma\rho}{\delta}r_{t}
+\frac{\theta\alpha\sigma\kappa\rho}{\delta}\int_{0}^{t}r_{s}ds}\right]
e^{-\frac{\theta\alpha\sigma\rho r_{0}}{\delta}
-\frac{\theta\alpha\sigma\kappa\gamma\rho}{\delta}t
+\frac{1}{2}\theta^{2}\alpha^{2}\sigma^{2}(1-\rho^{2})t},\nonumber
\end{align}
where the last line uses the fact that
$W_{t}=\frac{r_{t}-r_{0}}{\delta}-\frac{\kappa\gamma}{\delta}t+
\frac{\kappa}{\delta}\int_{0}^{t}r_{s}ds$.
\end{proof}
Let $u(t,r):=E_{r_{0}=r}\left[e^{
\theta(1-\alpha)\int_{0}^{t}r_{s}ds+\frac{\theta\alpha\sigma\rho}{\delta}r_{t}
+\frac{\theta\alpha\sigma\kappa\rho}{\delta}\int_{0}^{t}r_{s}ds}\right]$.
Then, $u(t,r)$ satisfies the following partial differential equation,
\begin{equation}
\begin{cases}
\frac{\partial u}{\partial t}
=\kappa(\gamma-r)\frac{\partial u}{\partial r}+\frac{1}{2}\delta^{2}
\frac{\partial u^{2}}{\partial r^{2}}+\left(\theta(1-\alpha)
+\frac{\theta\alpha\sigma\kappa\rho}{\delta}\right)ru=0,
\\
u(0,r)=e^{\frac{\theta\alpha\sigma\rho}{\delta}r}.
\end{cases}
\end{equation}
Let us try $u(t,r)=e^{A(t)+B(t)r}$. Then, we get
\begin{equation}
\begin{cases}
A'(t)=\kappa\gamma B(t)+\frac{1}{2}\delta^{2}B(t)^{2},
\\
B'(t)=-\kappa B(t)+\theta(1-\alpha)+\frac{\theta\alpha\sigma\kappa\rho}{\delta},
\\
A(0)=0,\quad B(0)=\frac{\theta\alpha\sigma\rho}{\delta}.
\end{cases}
\end{equation}
It is not hard to see that $B(t)\rightarrow\frac{\theta(1-\alpha)}{\kappa}+\frac{\theta\alpha\sigma\rho}{\delta}$
as $t\rightarrow\infty$ and therefore
\begin{align}
\frac{A(t)}{t}
&=\frac{\kappa\gamma}{t}\int_{0}^{t}B_{s}ds
+\frac{\delta^{2}}{2t}\int_{0}^{t}B(s)^{2}ds
\\
&\rightarrow\kappa\gamma\left(\frac{\theta(1-\alpha)}{\kappa}
+\frac{\theta\alpha\sigma\rho}{\delta}\right)
+\frac{\delta^{2}}{2}\left(\frac{\theta(1-\alpha)}{\kappa}
+\frac{\theta\alpha\sigma\rho}{\delta}\right)^{2},
\nonumber
\end{align}
as $t\rightarrow\infty$. Hence, we conclude that
\begin{align}
&\lim_{t\rightarrow\infty}\frac{1}{t}\log
E_{r_{0}=r}
\left[e^{\theta(\alpha\sigma B_{t}+(1-\alpha)\int_{0}^{t}r_{s}ds)}\right]
\\
&=\kappa\gamma\left(\frac{\theta(1-\alpha)}{\kappa}
+\frac{\theta\alpha\sigma\rho}{\delta}\right)
+\frac{\delta^{2}}{2}\left(\frac{\theta(1-\alpha)}{\kappa}
+\frac{\theta\alpha\sigma\rho}{\delta}\right)^{2}
\nonumber
\\
&\qquad\qquad\qquad
-\frac{\theta\alpha\sigma\kappa\gamma\rho}{\delta}
+\frac{1}{2}\theta^{2}\alpha^{2}\sigma^{2}(1-\rho^{2})
.\nonumber
\end{align}

\begin{theorem}
\begin{align}
\Lambda(\alpha)&=\lim_{t\rightarrow\infty}\frac{1}{t}\log E[u(V_{t})]
\\
&=\kappa\gamma\left(\frac{\theta(1-\alpha)}{\kappa}
+\frac{\theta\alpha\sigma\rho}{\delta}\right)
+\frac{\delta^{2}}{2}\left(\frac{\theta(1-\alpha)}{\kappa}
+\frac{\theta\alpha\sigma\rho}{\delta}\right)^{2}
\nonumber
\\
&\qquad\qquad
-\frac{\theta\alpha\sigma\kappa\gamma\rho}{\delta}
+\frac{1}{2}\theta^{2}\alpha^{2}\sigma^{2}(1-\rho^{2})
+\theta\alpha\mu-\frac{1}{2}\theta\alpha^{2}\sigma^{2}.
\nonumber
\end{align}
When $\frac{\delta^{2}\theta}{2\kappa^{2}}
-\frac{\delta\theta\sigma\rho}{\kappa}
+\frac{\sigma^{2}\theta}{2}-\frac{\sigma^{2}}{2}\geq 0$,
the optimal $\alpha^{\ast}$ is given by
\begin{equation}
\alpha^{\ast}=
\begin{cases}
0 &\text{if 
$\gamma+\frac{\delta^{2}\theta}{2\kappa^{2}}
\geq\frac{1}{2}(\theta-1)\sigma^{2}+\mu$},
\\
1 &\text{if $\gamma+\frac{\delta^{2}\theta}{2\kappa^{2}}
<\frac{1}{2}(\theta-1)\sigma^{2}+\mu$}.
\end{cases}
\end{equation}
Otherwise, the optimal $\alpha^{\ast}$ is given by
\begin{equation}
\alpha^{\ast}
=
\begin{cases}
1 &\text{if $\alpha^{\dagger}\geq 1$},
\\
\alpha^{\dagger} &\text{if $\alpha^{\dagger}\in(0,1)$},
\\
0 &\text{if $\alpha^{\dagger}\leq 0$},
\end{cases}
\end{equation}
where
\begin{equation}
\alpha^{\dagger}=
\frac{-\gamma\theta+\theta\mu+\frac{\delta\theta^{2}\sigma\rho}{\kappa}
-\frac{\delta^{2}\theta^{2}}{\kappa^{2}}}
{2\theta\left[-\frac{\delta^{2}\theta}{2\kappa^{2}}
+\frac{\delta\theta\sigma\rho}{\kappa}
-\frac{\sigma^{2}\theta}{2}+\frac{\sigma^{2}}{2}\right]}.
\end{equation}
\end{theorem}

\begin{proof}
Recall that $V_{t}=V_{0}e^{\alpha\mu t+\alpha\sigma B_{t}
-\frac{1}{2}\alpha^{2}\sigma^{2}t+(1-\alpha)\int_{0}^{t}r_{s}ds}$.
Applying Lemma \ref{VLemma}, we have
\begin{align}
\Lambda(\alpha)&=\lim_{t\rightarrow\infty}\frac{1}{t}\log E[u(V_{t})]
\\
&=\kappa\gamma\left(\frac{\theta(1-\alpha)}{\kappa}
+\frac{\theta\alpha\sigma\rho}{\delta}\right)
+\frac{\delta^{2}}{2}\left(\frac{\theta(1-\alpha)}{\kappa}
+\frac{\theta\alpha\sigma\rho}{\delta}\right)^{2}
\nonumber
\\
&\qquad\qquad
-\frac{\theta\alpha\sigma\kappa\gamma\rho}{\delta}
+\frac{1}{2}\theta^{2}\alpha^{2}\sigma^{2}(1-\rho^{2})
+\theta\alpha\mu-\frac{1}{2}\theta\alpha^{2}\sigma^{2}.
\nonumber
\end{align}
Thus, $\Lambda(\alpha)$
is a quadratic function of $\alpha$. 
If the coefficient of $\alpha^{2}$ in $\Lambda(\alpha)$ is non-negative,
i.e.
\begin{equation}
\theta\left[\frac{\delta^{2}\theta}{2\kappa^{2}}
-\frac{\delta\theta\sigma\rho}{\kappa}
+\frac{\sigma^{2}\theta}{2}-\frac{\sigma^{2}}{2}\right]\geq 0,
\end{equation}
then $\Lambda(\alpha)$ is convex in $\alpha$ and 
the optimal $\alpha$ is achieved at either $\alpha=0$
or $\alpha=1$. Indeed, one can compute that
\begin{equation}
\Lambda(0)=\gamma\theta+\frac{\delta^{2}\theta^{2}}{2\kappa^{2}},
\end{equation}
and
\begin{equation}
\Lambda(1)=\frac{1}{2}(\theta^{2}-\theta)\sigma^{2}+\theta\mu.
\end{equation}
Therefore, when $\frac{\delta^{2}\theta}{2\kappa^{2}}
-\frac{\delta\theta\sigma\rho}{\kappa}
+\frac{\sigma^{2}\theta}{2}-\frac{\sigma^{2}}{2}\geq 0$,
\begin{equation}
\alpha^{\ast}=
\begin{cases}
0 &\text{if 
$\gamma+\frac{\delta^{2}\theta}{2\kappa^{2}}
\geq\frac{1}{2}(\theta-1)\sigma^{2}+\mu$},
\\
1 &\text{if $\gamma+\frac{\delta^{2}\theta}{2\kappa^{2}}
<\frac{1}{2}(\theta-1)\sigma^{2}+\mu$}.
\end{cases}
\end{equation}

If the coefficient of $\alpha^{2}$ in $\Lambda(\alpha)$ is negative,
the function $\Lambda(\alpha)$ has a unique maximum at some 
$\alpha^{\dagger}\in(0,\infty)$ and
\begin{equation}
\alpha^{\ast}
=
\begin{cases}
1 &\text{if $\alpha^{\dagger}\geq 1$},
\\
\alpha^{\dagger} &\text{if $\alpha^{\dagger}\in(0,1)$},
\\
0 &\text{if $\alpha^{\dagger}\leq 0$},
\end{cases}
\end{equation}
where
\begin{equation}
\alpha^{\dagger}=
\frac{-\gamma\theta+\theta\mu+\frac{\delta\theta^{2}\sigma\rho}{\kappa}
-\frac{\delta^{2}\theta^{2}}{\kappa^{2}}}
{2\theta\left[-\frac{\delta^{2}\theta}{2\kappa^{2}}
+\frac{\delta\theta\sigma\rho}{\kappa}
-\frac{\sigma^{2}\theta}{2}+\frac{\sigma^{2}}{2}\right]}.
\end{equation}
\end{proof}

\section{Concluding Remarks}

In this article, we studied the optimal long-term strategy for a static investor for the Heston model, the 3/2 model, the jump diffusion
model and the Black-Scholes model with Vasicek interest rate.
It will be interesting to generalize our results 
to the multivariate case, i.e. when the investor can invest in a basket of stocks $S^{(i)}_{t}$, $1\leq i\leq d$,
and the wealth process is given by
\begin{equation}
\frac{dV_{t}}{V_{t}}=\sum_{i=1}^{d}\alpha_{i}\frac{dS^{(i)}_{t}}{S^{(i)}_{t}}+\left(1-\sum_{i=1}^{d}\alpha_{i}\right)r_{t}dt.
\end{equation}
One can also study the Heston model, the 3/2 model, and the jump diffusion model with stochastic interest rate.
In Section \ref{VSection}, the interest rate is assumed to follow the Vasicek model. A drawback of the Vasicek model
is that the process can go negative with positive probability. Our analysis in Section \ref{VSection} cannot be directly
applied to the Cox-Ingersoll-Ross interest rate unless one assumes that $\rho=0$. This can be left for the future investigations.

\section*{Acknowledgements}

The author also wishes to thank an annonymous referee for very careful readings of the manuscript and helpful suggestions
that greatly improved the paper.
The author is supported by NSF grant DMS-0904701, DARPA grant and MacCracken Fellowship at New York University.

\end{document}